\documentclass[11pt]{article}
\usepackage[margin=1.2in]{geometry}
\usepackage{times}
\usepackage{helvet}
\usepackage{courier}
\usepackage{float}
\usepackage{subfig}
\usepackage{graphicx}
\usepackage{theorem}
\usepackage{amsmath,amssymb}
\usepackage{url}
\usepackage{epsfig}
\usepackage{color}
\usepackage{bbm}
\usepackage{epsfig}
\usepackage{epstopdf}
\usepackage[linesnumbered,ruled,vlined, algo2e]{algorithm2e}
\usepackage{algpseudocode}
\usepackage{algorithmicx}
\usepackage{algorithm}

\newtheorem{theorem}{Theorem}
\newtheorem{lemma}{Lemma}

\newtheorem{corollary}{Corollary}

\newtheorem{definition}{Definition}

\newcommand{\eps}{\epsilon}

\newcommand{\now}{\mathtt{now}}
\newcommand{\OPT}{\mathtt{OPT}}
\newcommand{\ALG}{\mathcal{A}}

\newcommand{\sieve}{{\small\tt SieveStream}}

\newcommand{\swrd}{{\small\tt SW-RD}}
\newcommand{\snaive}{{\small\tt SieveNaive}}
\newcommand{\sgreedy}{{\small\tt SieveGreedy}}
\newcommand{\random}{{\small\tt Random}}
\newcommand{\greedy}{{\small\tt Greedy}}
\newcommand{\eye}{{\small\tt Eyes}}
\newcommand{\gas}{{\small\tt GasSensor}}
\newcommand{\wc}{{\small\tt WorldCup}}

\renewcommand{\S}{\mathcal{S}}

\DeclareMathOperator*{\argmax}{argmax}

\renewcommand{\paragraph}[1]{\medskip \noindent {\bf #1.}}

\newenvironment{proof}{\trivlist\item[]\emph{Proof}:}%
{\unskip\nobreak\hskip 1em plus 1fil\nobreak$\Box$
\parfillskip=0pt%
\endtrivlist}

\renewcommand{\paragraph}[1]{\medskip \noindent {\bf #1.}}

\begin{document}

\title{Submodular Maximization over Sliding Windows}

\author{Jiecao Chen\\Indiana University Bloomington\\
  jiecchen@umail.iu.edu
\and  Huy L. Nguyen\\Northeastern University,\\
hlnguyen@cs.princeton.edu
\and Qin Zhang\\Indiana University Bloomington\\
qzhangcs@indiana.edu}

\date{}


\maketitle

\begin{abstract}
In this paper we study the extraction of representative elements in the data stream model in the form of submodular maximization.  Different from the previous work on streaming submodular maximization, we are interested only in the recent data, and study the maximization problem over sliding windows.   We provide a general reduction from the sliding window model to the standard streaming model, and thus our approach works for general constraints as long as there is a corresponding streaming algorithm in the standard streaming model. As a consequence, we obtain the first algorithms in the sliding window model for maximizing a monotone/non-monotone submodular function under cardinality and matroid constraints. We also propose several heuristics and show their efficiency in real-world datasets.
\end{abstract}

\section{Introduction}
\label{sec:intro}

The last few decades have witnessed an explosion in the amount of data involved in machine learning tasks. In many cases, the data volume exceeds our storage capacity and demands new techniques that can effectively learn while operating within stringent storage and computation limits. {\em Streaming algorithms}~\cite{Muthukrishnan05} have emerged as a powerful approach to cope with massive data volume. In this model, the learning algorithm processes the dataset one element at a time by a linear scan and keeps only a small summary of the dataset in the memory; it is then able to compute the objective function on the processed elements using the summary. Various popular techniques in machine learning operate in this model, such as stochastic gradient descent, perceptron, etc. This model is particularly useful when the dataset is too large to fit in the memory or if the data is generated in real time (e.g., in online learning).

A common issue in online learning/streaming data mining is that the underlying distribution that generates the data may change over time. Therefore we had better consider only the most recently data in the stream.
Streaming algorithms over sliding windows have already been designed for several problems, including $k$-median clustering~\cite{BDMO03}, kernel least square regression \cite{VVS06}, $k$-means and coreset construction~\cite{BLLM15},  etc.
  
The window size can again be very large and does not fit the main memory. The problem becomes more severe in the case when kernel method is applied to deal with the nonlinear system -- the resulting kernel matrix may need $O(W^2)$ memory where $W$ is the window size.  A natural idea to resolve this issue is to select representative data items from the window for computing the objective function.

Submodular functions, an abstract and broad class of functions, have recently become an attractive candidate for modeling a variety of scenarios in machine learning, from exemplar-based clustering~\cite{KG10}, summarization~\cite{SSSJ12} to determinantal point processes~\cite{GKT12}. In recent years there have been quite some work on designing streaming algorithm for optimizing submodular functions~\cite{KG10,KMVV15,BMKK14,CK15,CGQ15}. However, we are not aware of any previous work dealing with streaming data over sliding windows in the context of submodular optimization. 

In this work, we present a general reduction from the sliding window model to the standard streaming model. As a consequence, we immediately obtain algorithms in the sliding window model by combining with previous works in the standard streaming model for cardinality constraints~\cite{KG10,KMVV15,BMKK14}, matroid and matching constraints~\cite{CK15}, and non-monotone functions~\cite{CGQ15}.  We also propose a few heuristics and compare their performance on real-world datasets.

\newcommand{\calI}{\mathcal{I}}
\section{Preliminaries}
\label{sec:pre}
Let $V$ be the finite  ground set (possibly multiset). Let $f:2^V \rightarrow \mathbb{R}$ be a function mapping subsets of $V$ to numbers. We say $f$ is \emph{submodular} if $f(A \cup \{v\}) - f(A) \geq f(B \cup \{v\}) - f(B)$ for any $A \subseteq B \subset V$, and $v \in V\backslash B$.   We define $f(v|A) = f(A\cup \{v\}) - f(A)$ as the \emph{marginal gain} of $v$ given $A$.  If further we have $f(A) \leq f(B)$ for any $A \subseteq B \subseteq V$, we say $f$ is \emph{monotone}.

In the general form of \emph{constrained submodular maximization} problem, we consider solving
\begin{equation}
  \label{eq:general}
 \argmax_{S\in \calI} f(S),
\end{equation}
where $\calI$ is a collection of subsets of $V$ which we call the \emph{constraint}. In particular, when $\calI = \{ S \subseteq V ~|~ |S| \leq k\}$, Expression (\ref{eq:general}) is known as the \emph{cardinality constrained submodular maximization} problem. Definitions for other constraints can be found in e.g. \cite{B14}.  We say constraint $\calI$ \emph{hereditary} if $A \in \calI$ implies that any subset of $A$ is also in $\calI$.

For a constrained submodular maximization problem, we use $\OPT$ to represent the solution of (\ref{eq:general}).  W.l.o.g.\ we assume $1 \le f(\OPT) \le M$ for a parameter $M$.

In the streaming model, we consider the ground set $V$ as an ordered sequence of {\em items} $e_1, e_2, \ldots, e_n$, each consumes one word of space. Each index is called a \emph{time step}.  In the sliding window model, we specify the window size $W$. At any time step, we are only interested in the most recent $W$ items which defines the ground set $V$ at that moment. Note that when $W \rightarrow \infty$, $V$ becomes the set of all received items; we thus also call the standard streaming model  the infinite window model. For two streams/sequences $S_1$ and $S_2$, let $S_1\circ S_2$ be their concatenation, which is $S_1$ followed by $S_2$.

Let $\ALG$ be an algorithm solving the constrained submodular maximization problem. We use $\ALG(V)$ to represent the function value of the solution returned by $\ALG$ operating on the ground set $V$.  


\section{Algorithms}
\label{sec:algo}

\subsection{A Reduction to the Infinite Window Case}
\label{sec:reduction}

In this section we design sliding window algorithms by a generic reduction to standard infinite window streaming algorithms.  Our reduction can work with any standard streaming algorithms that satisfy some natural properties. Throughout this section, we assume that the constraint in the submodular maximization problem is hereditary (see the preliminaries for definition).

\begin{definition}
\label{def:good-algo}

Let $\ALG$ be an algorithm operating on a stream of elements, and $\ALG(S)$ be the value of the solution found by $\ALG$ on the stream $S$. We say $\ALG$ is $c$-compliant if it satisfies the following conditions:
\begin{itemize}
\item {\em (Monotonicity)} If $S_1$ is a prefix of $S_2$ then $\ALG(S_1) \le \ALG(S_2)$.
\item {\em ($c$-Approximation)} For any stream $S$, let $\OPT$ be the optimal solution. Then $\ALG(S) \ge c\cdot f(\OPT)$.
\end{itemize}
\end{definition}

The following lemma is similar to the smooth histogram technique introduced in \cite{BO07}. However,  smooth histogram only works for problems with good approximations. More precisely, it is assumed that there is an algorithm with $(1-\eps)$ approximation where $\eps < 1/4$, which does not hold for many submodular problems. Our algorithm (Algorithm \ref{alg:reduction}) also works for streaming algorithms \emph{without} the monotonicity property, though monotonicity allows us to get a better approximation. Any streaming algorithm can be modified to satisfy the monotonicity property at the cost of increased update time: after every update, the algorithm calculates a solution and it keeps track of the best solution over time.

\begin{lemma}
\label{lem:smooth}
Let stream $S_1\circ S_2 \circ S_3$ be the concatenation of three (sub)streams $S_1, S_2, S_3$. Given a $c$-compliant algorithm $\ALG$, if $\ALG(S_1\circ S_2) \le (1+\eps) \ALG(S_2)$, then $\ALG(S_2\circ S_3) \ge \frac{c}{2+\eps}\cdot f(\OPT_{123})$ where $\OPT_{123}$ is the optimal solution for the stream $S_1\circ S_2 \circ S_3$.
\end{lemma}

\begin{proof}
Let $\OPT_{123}, \OPT_{23}, \OPT_{12}$ be the optimal solution for the streams $S_1\circ S_2\circ S_3, S_2\circ S_3, S_1\circ S_2$, respectively. We have 
\begin{equation}
\label{eq:b-1}
\textstyle \frac{1}{c} \cdot \ALG(S_2\circ S_3) \ge f(\OPT_{23}).
\end{equation}
We also have 
\begin{equation}
\label{eq:b-2}
\textstyle f(\OPT_{12}) \le \frac{1}{c} \cdot \ALG(S_1\circ S_2) \le \frac{1+\eps}{c} \cdot \ALG(S_2) \le \frac{1+\eps}{c} \cdot \ALG(S_2\circ S_3).
\end{equation}

Combining (\ref{eq:b-1}) and (\ref{eq:b-2}), we obtain
\begin{eqnarray*}
\textstyle \frac{2+\eps}{c} \cdot \ALG(S_2 \circ S_3) &\ge& f(\OPT_{12})+f(\OPT_{23}) \\ 
&\ge& f(\OPT_{123}\cap S_1)\\ & &+ f(\OPT_{123}\cap (S_2\circ S_3)) \\ &\ge& f(\OPT_{123}).
\end{eqnarray*}
\end{proof}

We can also show a similar lemma for algorithms satisfying the $c$-approximation property but not monotonicity.

\begin{lemma}
\label{lem:smooth2}
Let stream $S_1\circ S_2 \circ S_3$ be the concatenation of three (sub)streams $S_1, S_2, S_3$. Given an algorithm $\ALG$ with $c$-approximation property, if $\ALG(S_1\circ S_2) \le (1+\eps) \ALG(S_2)$, then $\ALG(S_2\circ S_3) \ge \frac{c^2}{c+1+\eps}\cdot f(\OPT_{123})$ where $\OPT_{123}$ is the optimal solution for the stream $S_1\circ S_2 \circ S_3$.
\end{lemma}
The proof is similar to that of Lemma \ref{lem:smooth2}. The major modification is to use the following inequality, which does not require the monotonicity property, as a replacement of (\ref{eq:b-2}),
\begin{eqnarray*}
\label{eq:b-4}
f(\OPT_{12}) &\le& \frac{1}{c} \cdot \ALG(S_1\circ S_2) \le \frac{1+\eps}{c} \cdot \ALG(S_2)  \\
            &\le& \frac{(1+\eps)}{c}f(\OPT_{23})  
            \le \frac{1+\eps}{c^2} \ALG(S_2\circ S_3). 
\end{eqnarray*}



We have the following theorem.
\begin{theorem}
\label{thm:reduction}

There is an algorithm for constrained submodular maximization over sliding windows that achieves a $c/(2+\eps)$-approximation using $O(s/\eps \cdot \log M)$ space and $O(t/\eps \cdot \log M)$ update time per item, provided that there is a corresponding $c$-compliant streaming algorithm using $s$ space and $t$ update time per item.
\end{theorem}

\begin{algorithm2e}[t]
\DontPrintSemicolon 
\KwIn{$M$: an upper bound of $f(\OPT)$ over sliding windows.  $W$: the size of the window. $\ALG$ is an infinite window streaming algorithm which we use as a blackbox.}
\ForEach{new incoming element $e_i$}{
	start a new instance $\ALG^{(i)}$\;
	drop all maintained instances $\ALG^{(y)}$ where $y \le i - W$\;
	update all the remaining instances, denoted by $\ALG^{(t_1)}, \ldots, \ALG^{(t_u)}$, with $e_i$\;
	$j \gets 1$\;
	\While{$j < u$}{
                $x \gets u$\; \label{ln:a-1}
		\While{ $(1+\eps) \ALG^{(t_x)}(e_{t_x}, \ldots, e_i) < \ALG^{(t_j)}(e_{t_j}, \ldots, e_i)$}{
                  $x \gets x - 1$
                }	
                Prune all $\ALG^{(t_v)}$ with $j < v < x$\;  \label{ln:a-2}
                \If{$x \leq j$}{
                  $x \gets j + 1$
                }
                $j \gets x$
              }
}
\Return (at query) $\ALG^{(t_b)}(e_{t_b}, \ldots, e_{\now})$ where $t_b = \min \{t_j \in [\now - W + 1, \now]\ |\  \ALG^{(t_j)} \text{ is maintained}\}$
\caption{\swrd$(\mathcal{A}, W, M)$: A Reduction to Infinite Window Streaming Algorithms}
\label{alg:reduction}
\end{algorithm2e}

The pseudocode of the algorithm is described in Algorithm~\ref{alg:reduction}.  We now explain it in words.  The algorithm maintains a collection of instances of $\ALG$ starting at different times $t_1 < t_2 < \ldots < t_u$,  which we will denote by $\ALG^{(t_1)}, \ALG^{(t_2)}, \ldots, \ALG^{(t_u)}$. Upon receiving a new element $e_i$, we perform the following operations. We first create a new instance $\ALG^{(i)}$. Next, we drop those instances of $\ALG$ that are expired, and update all the remaining instances with the new element $e_i$. Finally we perform a pruning procedure: We start with $t_1$. Let $t_x$ be the maximum time step among all the maintained instances of $\ALG$ such that $(1+\eps) \ALG^{(t_x)}\ge \ALG^{(t_1)}$. We prune all the instances $\ALG^{(t_v)}$ where $1 < v < x$ (Line~\ref{ln:a-1}-\ref{ln:a-2}). We repeat this pruning procedure with (the next unpruned time step) $t_x$ and continue until we reach $t_u$. Note that the purpose of the pruning step is to make the remaining data stream satisfy $\ALG(S_1\circ S_2) \le (1+\eps) \ALG(S_2)$, so that Lemma \ref{lem:smooth} or Lemma \ref{lem:smooth2} applies.

The space usage of the algorithm is again easy to bound.  Note that after the pruning (we rename the remaining instances as $\ALG^{t_1}, \ALG^{t_2}, \ldots, \ALG^{t_u}$), for each $j = 1, \ldots, u-2$, we have $\ALG^{(t_j)} > (1+\eps) \ALG^{(t_{j+2})}$. Thus the number of instances of $\ALG$ is bounded by $O(1/\eps \cdot \log M)$ at all time steps.  When doing the pruning, we only need to calculate the value of each instance once, so the processing time per item is $O(t/\eps \cdot \log M)$.  

 We next give the proof of the correctness. Let us consider a window $[i - W + 1, i]$, and let $t_b = \min \{t_j \in [i - W + 1, i]\ |\ j = 1, \ldots, u \}$.  For this window we will report whatever $\ALG^{(t_b)}$ reports.   It is easy to see that if $t_b = i - W + 1$, then the algorithm is obviously correct since $\ALG^{(t_b)}$ is a streaming algorithm starting from time $t_b$.  We next consider the case when $t_b > i - W + 1$.  Let $t_c$ be the time step when the last $\ALG^{(t)}$ in $\{\ALG^{(i - W + 1)} \ldots, \ALG^{(t_b-1)}\}$ was pruned, and let $t_a < i - W  +1$ be the largest time step before $i-W+1$ such that $\ALG^{(t_a)}$ is active after the pruning step at time $t_c$. Note that $t_a$ must exist because pruning always happens between two active instances (at Line 9 of the algorithm, we prune between $j$ and $x$ exclusively). It is clear that $t_a < t_b \le t_c$. Let $S_1 = (e_{t_a}, \ldots, e_{t_b-1})$, $S_2 = (e_{t_b}, \ldots, e_{t_c})$, and $S_3 = (e_{t_c+1}, \ldots, e_i)$. By the pruning rule, we have $(1+\eps) \ALG(S_2) \ge \ALG(S_1\circ S_2)$.
Plugging in Lemma~\ref{lem:smooth}, we have 
\begin{equation}
 \label{eq:use-lemma}
\textstyle \ALG(S_2\circ S_3) \ge \frac{c}{2+\eps} \cdot f(\OPT_{123}),
\end{equation} 
   where $\OPT_{123}$ is the optimal solution for the stream $S_1\circ S_2 \circ S_3$, which includes the window $[i-W+1, i]$.
   
\smallskip

In \cite{BMKK14} the authors gave a $(1/2-\eps)$-compliant algorithm in the standard streaming model for monotone submodular maximization subject to cardinality constraint $k$, using $O(k \log k / \eps)$ space and $O(\log k/\eps)$ update time per item.  We thus have the following corollary.

\begin{corollary}
\label{cor:reduction-kdd}

There is an algorithm for monotone submodular maximization with a cardinality constraint over sliding windows that achieves a $(1/4 -\eps)$-approximation using $O(k \log k /\eps^2 \cdot \log M)$ words of space and $O(\log k /\eps^2 \cdot \log M)$ update time per item, where $k$ is the cardinality constraint.
\end{corollary}

If we drop the requirement of monotonicity, we have the following result.  The proof is the same as that for Theorem \ref{thm:reduction}, but uses Lemma \ref{lem:smooth2} instead of Lemma \ref{lem:smooth} in (\ref{eq:use-lemma}).
\begin{theorem}
\label{thm:reduction2}
There is an algorithm for constrained submodular maximization over sliding windows that achieves a $c^2/(c + 1 +\eps)$-approximation using $O(s/\eps \cdot \log M)$ space and $O(t/\eps \cdot \log M)$ update time per item, provided that there is a corresponding $c$-approximation streaming algorithm that uses $s$ space and $t$ update time per item.
\end{theorem}

In~\cite{CK15}, the authors gave a $1/(4p)$-approximation algorithm in the standard streaming model for monotone submodular maximization subject to $p$-matroid constraints using $O(k)$ space, where $k$ is the maximum rank of the $p$-matroids.  We thus have:

\begin{corollary}
There is an algorithm for monotone submodular maximization subject to $p$-matroid constraints over sliding windows that achieves a $1/(4p+(1+\eps)16p^2)$-approximation using $O(k /\eps \cdot \log M)$ words of space, where $k$ is the maximum rank of the $p$-matroids.
\end{corollary}

A similar result can be obtained by plugging the deterministic approximation algorithm in \cite{CGQ15}.

%
%

\subsection{Heuristic Algorithms}
\label{sec:heuristic}
In this section, we introduce two heuristic algorithms based on \sieve\ proposed by \cite{BMKK14}. \sieve\ achieves a $(1/2 - \eps)$-approximation for cardinality constrained monotone submodular maximization in the standard streaming model. 

We briefly review how \sieve\ works. To simplify the description, we assume an upper bound of $f(\OPT)$ (denote as $M$) is given. \cite{BMKK14} also shows how one can get rid of this assumption by estimating $f(\OPT)$ on the fly.  The algorithm works as follows: we guess in parallel the thresholds $T = (1 + \eps)^0, (1 + \eps)^1, \ldots, (1 + \eps)^L$ where $L = \log_{1+\eps}M = O(\frac{\log M}{\eps})$. For each fixed $T$ we maintain a buffer $S$ as the solution over the data stream.  Upon receiving a new item $e_i$, we add it to the buffer if the cardinality constraint has not yet been violated (i.e. $|S| < k$) and the marginal gain $f(e_i | S) > (T / 2 - f(S)) / (k - |S|)$. \cite{BMKK14} shows that as long as $(1-\eps) f(\OPT) \le T \le f(\OPT)$, the corresponding $S$ satisfies $f(S) \geq (1-\eps) f(\OPT) / 2$. So we simply return the best among all buffers.

The first heuristic algorithm \snaive\ is very simple. For each threshold $T$ and its associated $S$ in \sieve, upon receiving a new item $e_i$, we first drop the expired item (if any). All other steps are exactly the same as \sieve. 

The second heuristic \sgreedy\ is a hybrid of \sieve\ and the standard greedy algorithm \greedy\ \cite{NWF78}. Let $c > 0$ be a parameter and $W$ be the window size. We maintain $B$ as a buffer of samples over the sliding window. Upon receiving a new item $e_i$, we add $e_i$ to $B$ with probability $c/W$, and drop expired item (if any) from $B$.  On the other hand, we maintain an instance of \sieve\ with the following modification: whenever an item $e$ in a buffer $S$ (associated with a certain $T$) expired, we update $S$ by using \greedy\ to choose a solution of size $(|S| - 1)$ from $B\cup S\backslash\{e\}$.

The pseudocodes of the two heuristics are presented in Algorithm \ref{alg:snaive} and Algorithm \ref{alg:sgreedy} respectively.

\begin{algorithm2e}[t]
\DontPrintSemicolon 
\KwIn{$k$ the cardinality constraint; $W$: the size of the window; $M$ an upper bound of  }
$L \gets\log_{1+\eps}M$\;
\ForEach{$T = (1+\eps)^0, (1 + \eps)^1, \ldots, (1 + \eps)^L$}{
  $S_T \gets \emptyset$
}
\ForEach{new incoming element $e_i$}{
  \ForEach{$T = (1+\eps)^0, (1 + \eps)^1, \ldots, (1 + \eps)^L$}{
    Drop expired item (if any) from $S_T$\;
    \If{$|S_T| < k$ and $f(e_i | S_t) > (T / 2 - f(S_T)) / (k - |S_T|)$}{
      $S_T \gets S_T \cup \{e_i\}$
    }
  }
}

\Return{(at query) $\argmax_{S_T}f(S_T)$}
\caption{\snaive($k, W, M$)}
\label{alg:snaive}
\end{algorithm2e}

\begin{algorithm2e}[t]
\DontPrintSemicolon 
\KwIn{$k$ the cardinality constraint; $W$: the size of the window; $M$ an upper bound of $f(\OPT)$; $c$ parameter to control the sample probability}
$L \gets\log_{1+\eps}M$\;
\ForEach{$T = (1+\eps)^0, (1 + \eps)^1, \ldots, (1 + \eps)^L$}{
  $S_T \gets \emptyset$
}
$B\gets \emptyset$\;
\ForEach{new incoming element $e_i$}{
  Add $e_i$ to $B$ with probability $\frac{c}{W}$\;
  Drop expired item (if any) from $B$\;
  \ForEach{$T = (1+\eps)^0, (1 + \eps)^1, \ldots, (1 + \eps)^L$}{
  \If{there exists an expired item $e$ in $S_T$}{
       $S_T \gets$  output of running \greedy\ on $B\cup S_T\backslash\{e\}$ with cardinality constraint $(|S_T| - 1)$\;
    }

    \If{$|S_T| < k$ and $f(e_i | S_T) > (T / 2 - f(S_T)) / (k - |S_T|)$}{
      $S_T \gets S_T \cup \{e_i\}$
    }
  }
}

\Return{(at query) $\argmax_{S_T}f(S_T)$}
\caption{\sgreedy($k, W, c, M$)}
\label{alg:sgreedy}
\end{algorithm2e}

\newcommand{\calK}{\mathcal{K}}
\newcommand{\bbR}{\mathbb{R}}

\section{Applications}
The class of submodular functions contains a broad range of useful functions. Here we discuss two examples that have been used extensively in operations research, machine learning, and data mining. The performance of our algorithms in these settings is discussed in the experiments section.

\subsection{Maximum Coverage}
Let $\S = \{S_1, S_2, \ldots, S_n\}$  be a collection of subsets of $[M] = \{1, 2, \ldots, M\}$. In the  Maximum Coverage problem, we want to find \emph{at most} $k$ sets from $\S$ such that the cardinality of their union can be maximized.  More precisely, we define the utility function as  $f(\S') = |\cup_{S\in \S'}S|$, where $\S'$ is a subset of $\S$. It is straightforward to verify that the utility function defined is monotone submodular. The Maximum Coverage problem is a classical optimization problem and it is NP-Hard. We can formulate it using our notations as 
$\argmax_{\S' \subseteq \S,~|\S'| \leq k} f(\S').$

\subsection{Active Set Selection in Kernel Machines}
Kernel machines \cite{SS02} are powerful non-parametric learning techniques. They use kernels to reduce non-linear problems to linear tasks that have been well studied. The data set $V = \{x_1, x_2, \ldots, x_n\}$ is represented in a transformed space via the $n \times n$ kernel matrix $K_V$ whose $(i,j)$-th cell is $\calK(x_i, x_j)$
where $\calK: V\times V \rightarrow \bbR$ is the kernel function which is symmetric and positive definite. 

For large-scale problems, even representing the matrix $K_V$, which requires $O(n^2)$ space, is prohibited. The common practice is to select a small representative subset $S \subseteq V$ and only work with $K_S$. One popular way to measure the quality of selected set $S$ is to use \emph{Informative Vector Machine} (IVM) introduced by Laurence et al.\ \cite{LSH03}. Formally, we define $f: 2^V \rightarrow \bbR$ with $f(S) = \frac{1}{2} \log\det\left( \mathbf{I} + \sigma^{-2} K_S \right)$, where $\mathbf{I}$ is the identity matrix and $\sigma > 0$ is a parameter. IVM has a close connection to the entropy of  muti-variable Gaussian distribution \cite{B14}. It has been shown that $f$ is a monotone submodular function (see, e.g., \cite{B14}). We can then select the set $S\subset V$ by solving $\argmax_{S:|S|\leq k} f(S)$.


\section{Experiments}
\label{sec:exp}

In this section, we compare the following algorithms experimentally. We use the objective functions introduced in the previous section, and the dataset is fed as a data stream. We try to continuously maximize the objective functions on the most recent $W$ data points.
\begin{itemize}
\item \greedy: the standard greedy algorithm (c.f.\ \cite{NWF78}); does not apply to sliding windows. 
\item \sieve: the Sieve Streaming algorithm in \cite{BMKK14}; does not apply to sliding windows.
\item \snaive: Algorithm \ref{alg:snaive}  in this paper. 
\item \sgreedy: Algorithm \ref{alg:sgreedy} in this paper.
\item \random: random sampling over sliding windows \cite{BDM02} (i.e. maintain a random $k$ samples of elements in the sliding window at any time). 
\item \swrd: Algorithm~\ref{alg:reduction} in this paper, using \sieve\ as the $c$-compliant algorithm.
\end{itemize}
Note that neither \greedy\ nor \sieve\ can be used for submodular maximization over sliding windows. We thus have to run them in each selected window individually. If we want to continuously (i.e. for all sliding windows) report the solutions, then we need to initialize one instance of \sieve\ or \greedy\ for each window, which is space and time prohibitive. 

We run \greedy\ as it provides a benchmark of the qualities of solutions. We run \sieve\ in selected windows since \swrd\ uses it as a subroutine and we want to see how good the solutions of \swrd\ is compared with the original \sieve\ in practice.  

We have implemented all algorithms in C++ with the support of the C++ linear algebra library {\tt Armadillo} \cite{arma10}. 
 All experiments are conducted on a laptop equipped with an Intel Core i5 1.7GHz x 2 processor and 4GB RAM. The operating system is Linux Mint 17.2.


\paragraph{Datasets} We use three time-series datasets. 
\begin{itemize}
\item \eye \cite{eye13}: this dataset is from one continuous EEG measurement with the Emotiv EEG Neuroheadset. The duration of the measurement is 117 seconds. The dataset contains $14,980$ instances, each of which can be considered as a vector of dimension $15$.

\item \gas \cite{FSH+15}: this dataset contains the acquired time series from 16 chemical sensors exposed to gas mixtures at varying concentration levels.  Together with $3$ other features, each record can be considered as a vector of dimension $19$. There are $4,178,504$ records in total. We normalize the dataset first by column, and then by row.

\item \wc \cite{wc98}: this dataset contains all the requests made to the 1998 World Cup Web site on June 7, 1998. There are $5,734,309$ requests made on that day and we consider the requested resource URLs in each second as a set. This results in $24 \times 3600 = 86,400$ sets. 
\end{itemize}

\begin{figure}[!ht]
  \centering$
  \arraycolsep=0.2pt\def\arraystretch{1.0}
  \begin{array}{cc}
    \includegraphics[width=0.35\textwidth]{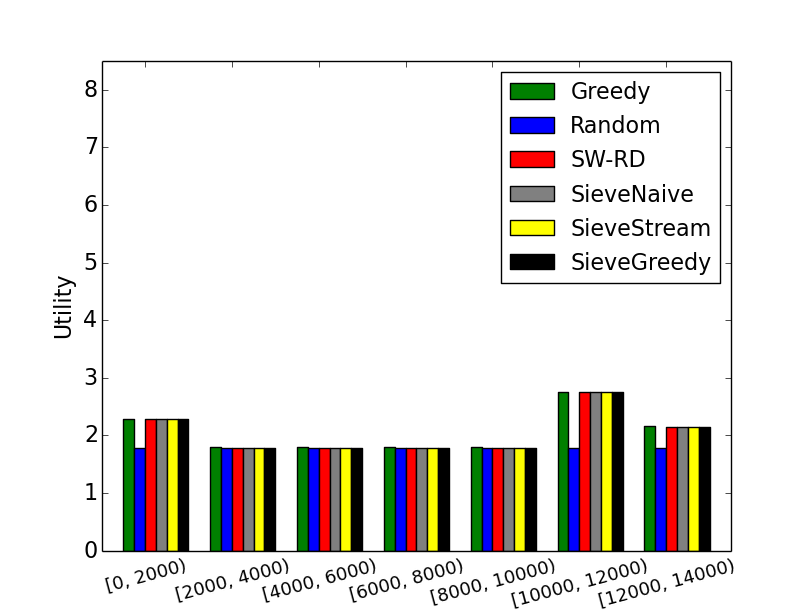}&
    \includegraphics[width=0.35\textwidth]{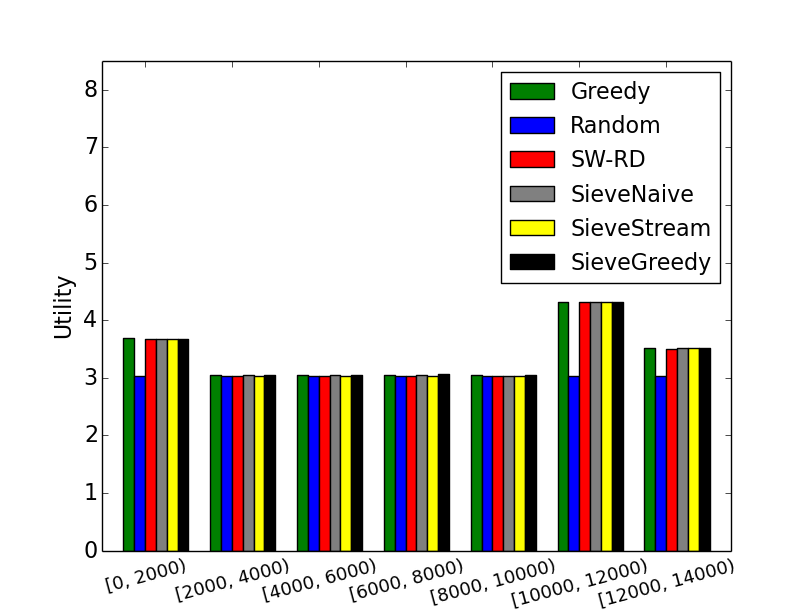}
\end{array}$
\caption{\eye\ dataset for active set selection; $k=5$ in the left figure, $k=20$ in the right;
  $W=2000$; $c = 20$ in \sgreedy; $x$-axis specifies the windows;  $y$-axis is the utility}
\label{fig:eye-acc}
\end{figure}

\begin{figure}[!ht]
  \centering$
  \arraycolsep=0.2pt\def\arraystretch{1.0}
  \begin{array}{cc}
    \includegraphics[width=0.35\textwidth]{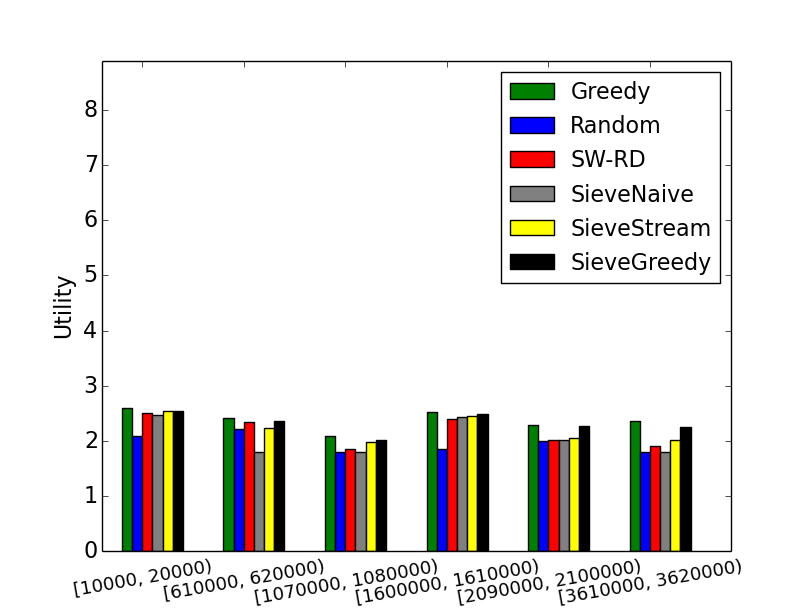}&
    \includegraphics[width=0.35\textwidth]{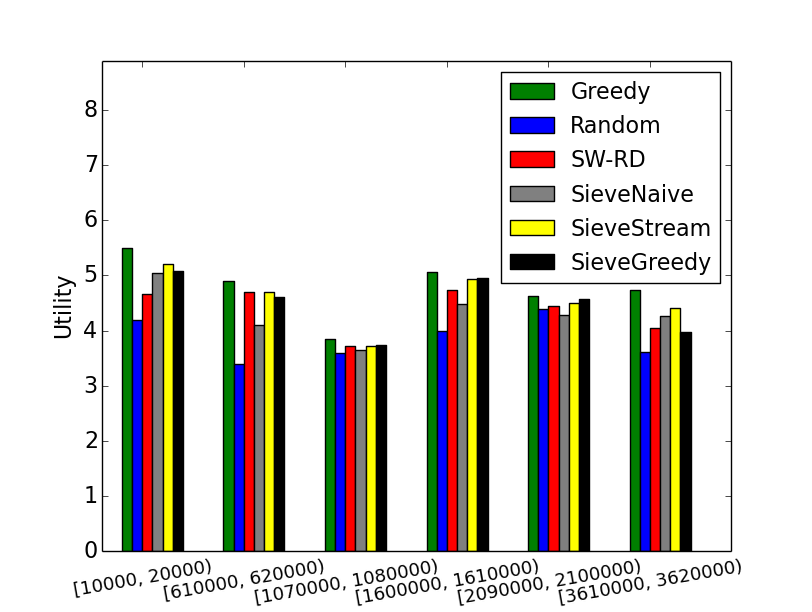}
  \end{array}$
\caption{\gas\ dataset for active set selection; $k=5$ in the left figure, $k=20$ in the right;
  $W=10000$; $c = 20$ in \sgreedy}
\label{fig:gas-acc}
\end{figure}

\begin{figure*}[!ht]
     \centering
     \subfloat[][$k=5$]{\includegraphics[width=0.35\textwidth]{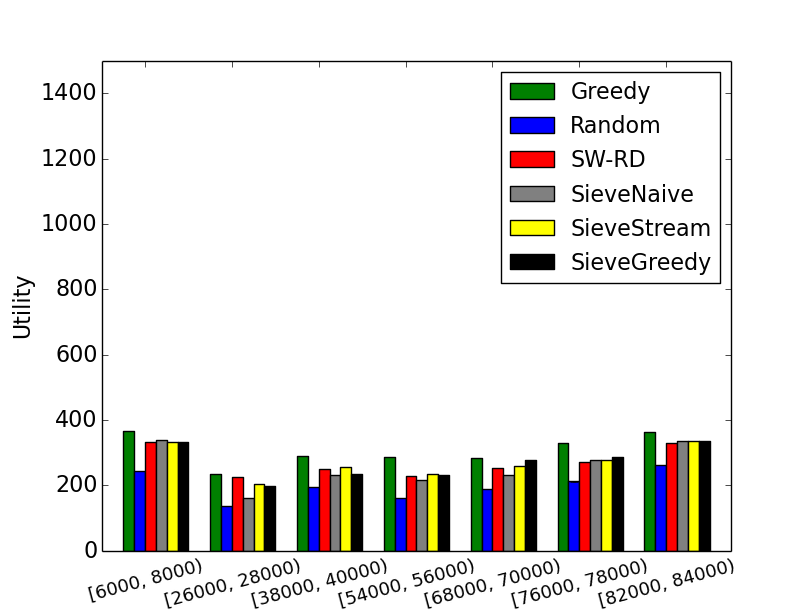}}
     \subfloat[][$k=20$]{\includegraphics[width=0.35\textwidth]{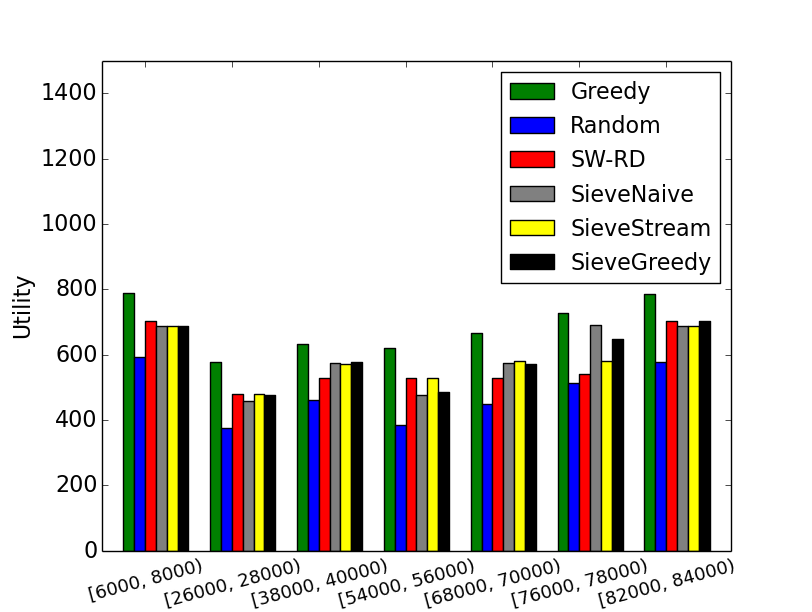}}     
     \subfloat[][\sgreedy,  $k=5$]{\includegraphics[width=0.35\textwidth]{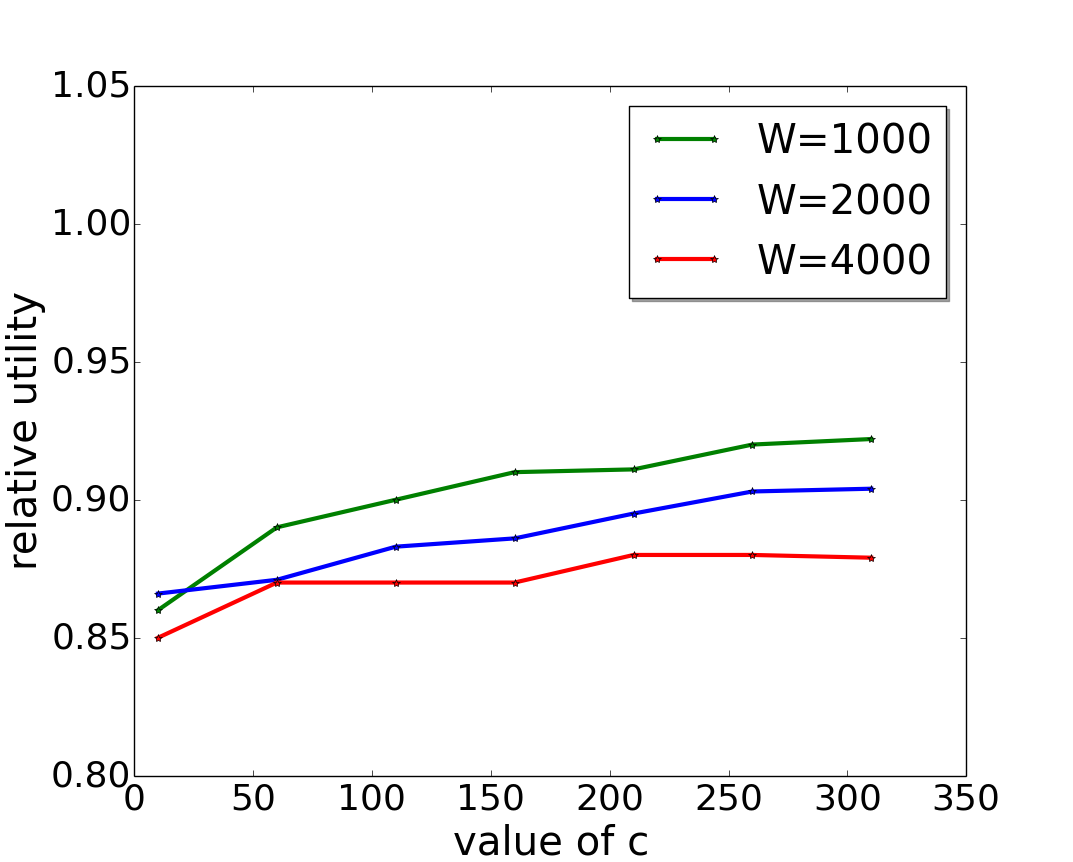}\label{fig:wc-sgreedy-c}}     
     \caption{\small \wc\ dataset for maximum coverage; 
     $W=2000$; $c = 20$ in \sgreedy\ except (\ref{fig:wc-sgreedy-c})
     }
     \label{fig:wc-acc}
\end{figure*}

\begin{figure*}[!ht]
     \centering
     \subfloat[][\swrd]{\includegraphics[width=0.35\textwidth]{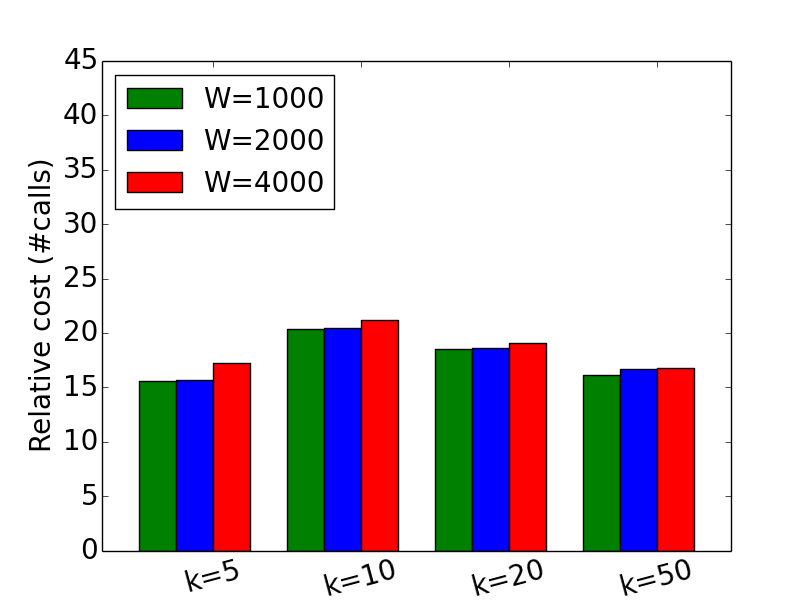}\label{fig:call-rd}}
     \subfloat[][\snaive]{\includegraphics[width=0.35\textwidth]{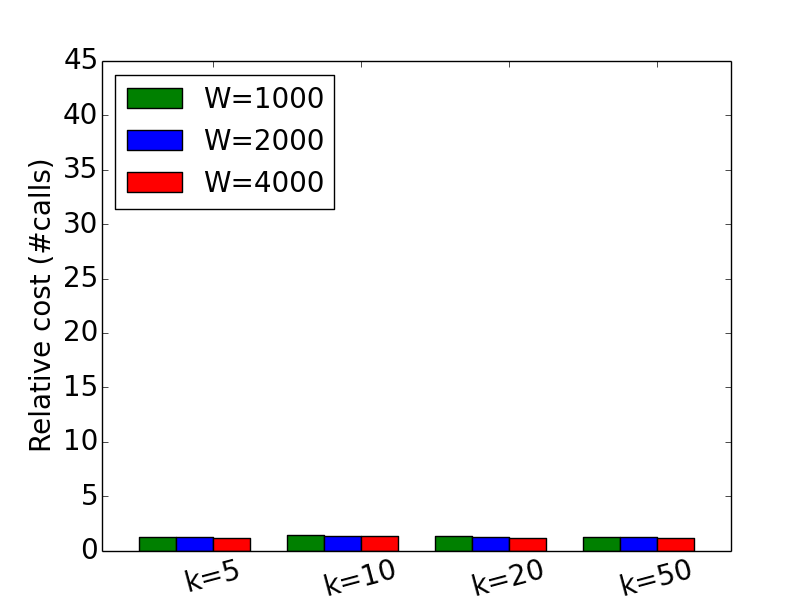}\label{fig:call-snaive}}\\
     \subfloat[][\sgreedy, $c=20$]{\includegraphics[width=0.35\textwidth]{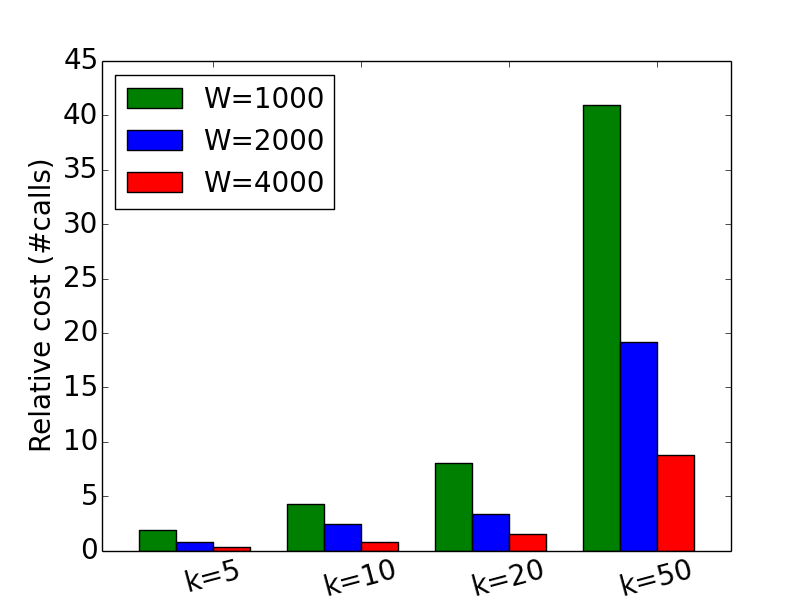}\label{fig:call-sgreedy}}
     \subfloat[][\sgreedy]{\includegraphics[width=0.35\textwidth]{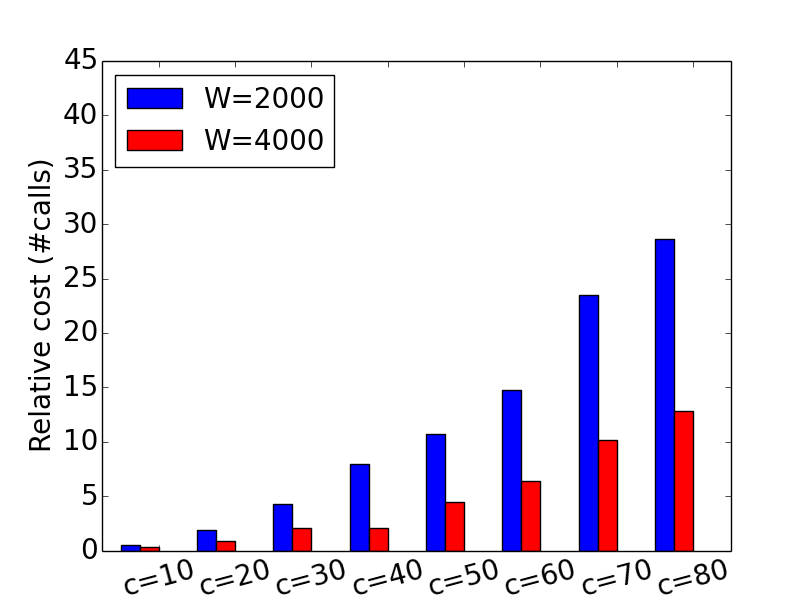}\label{fig:call-sgreedy-c}}
     \caption{\eye\ dataset for active set selection; \# function calls normalized by \sieve}\label{fig:eye-time}
\end{figure*}

\begin{figure*}[!ht]
     \centering
     \subfloat[][\sieve]{\includegraphics[width=0.35\textwidth]{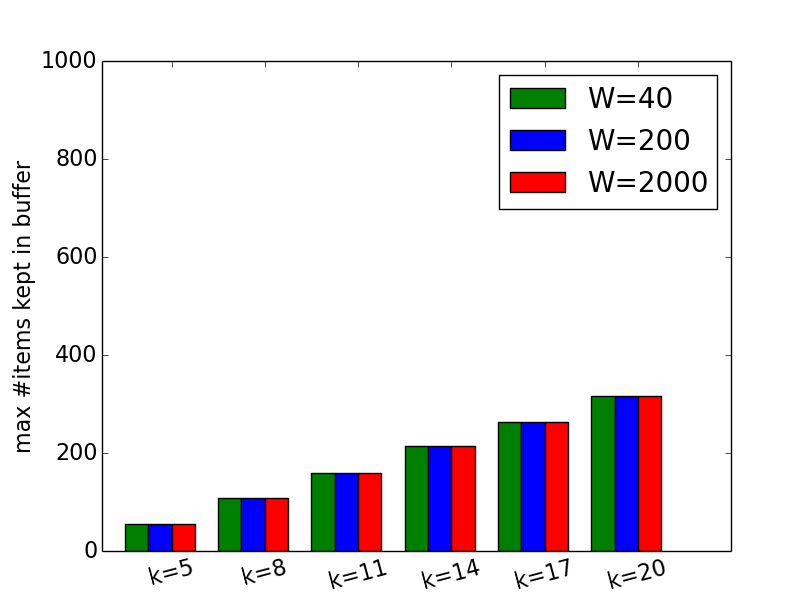}\label{fig:space-sieve}}
     \subfloat[][\swrd]{\includegraphics[width=0.35\textwidth]{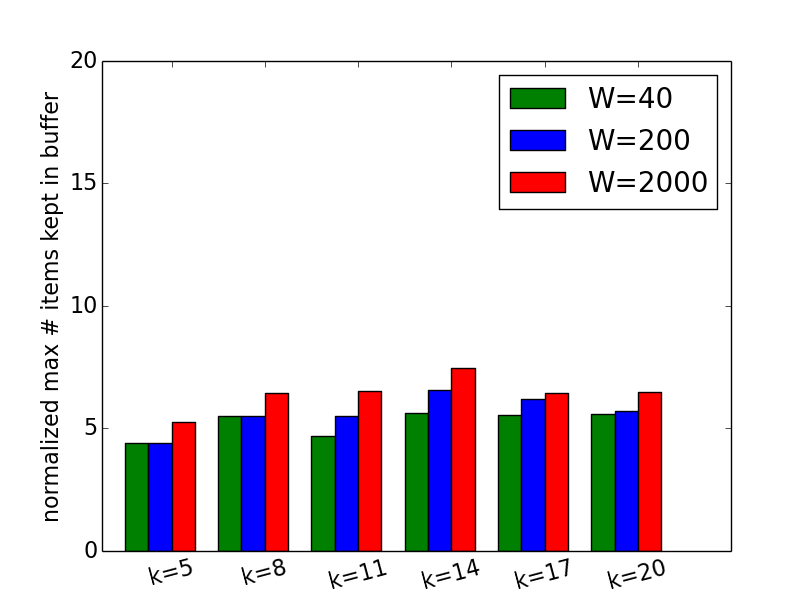}\label{fig:space-rd-nm}}\\
     \subfloat[][\snaive]{\includegraphics[width=0.35\textwidth]{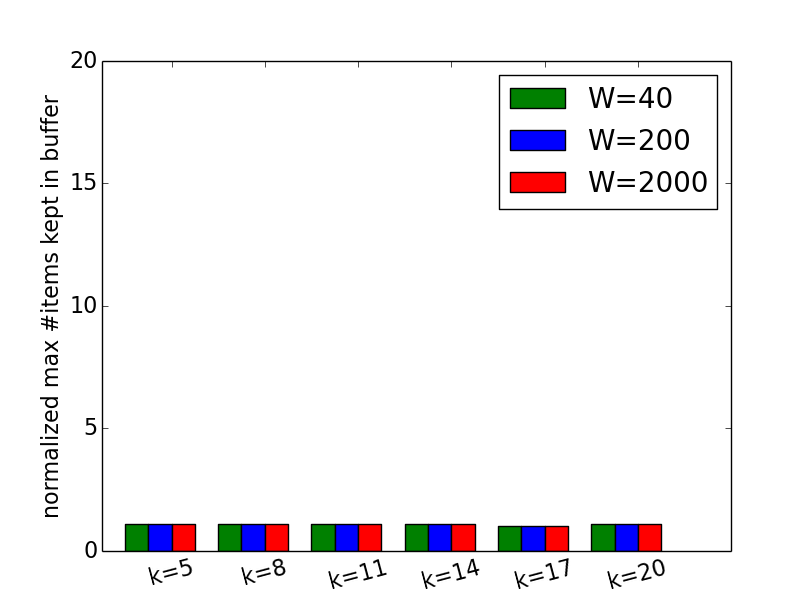}\label{fig:space-snaive-nm}}
     \subfloat[][\sgreedy, $c=20$]{\includegraphics[width=0.35\textwidth]{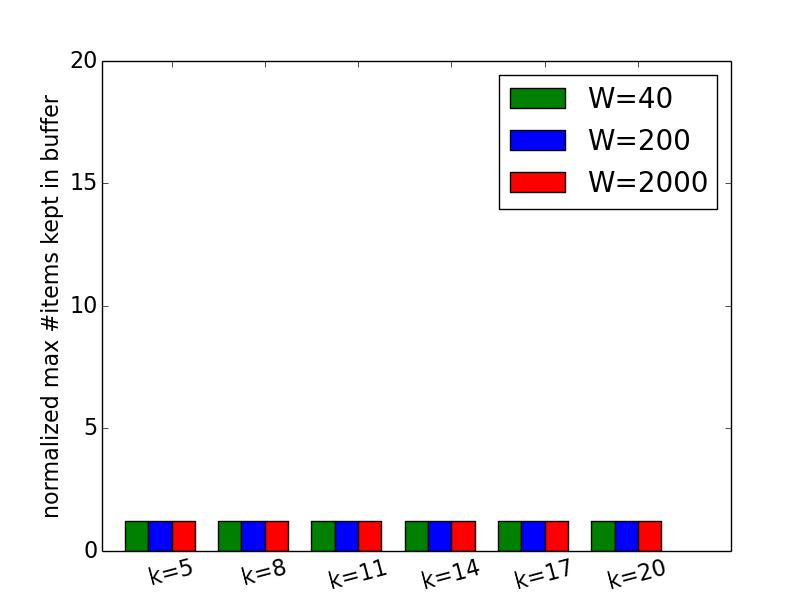}\label{fig:space-sgreedy-nm}}
     \caption{\eye\ dataset for active set selection; space usages measured by the peak number of items kept in the buffer; 
     (b), (c), and (d) are normalized by the space usages of \sieve}\label{fig:eye-space}
\end{figure*}


\paragraph{Discussion on the Results}
For the application of active set selection, we run experiments on both \eye\ and \gas\ datasets. We choose the squared exponential kernel as the kernel function: $\calK(x_i, x_j) = \exp(-\|x_i - x_j\|_2^2 / h^2)$;  we set $\sigma = 1$ and $h=0.75$. For the application of maximum coverage problem, we run experiments on the \wc\ dataset. For all algorithms, we set $\epsilon = 0.2$.  

It can be observed from Figure \ref{fig:eye-acc}, Figure \ref{fig:gas-acc} and Figure \ref{fig:wc-acc} that the maximum utility given by the standard greedy changes when we slide the window over the data stream. 
In Figure \ref{fig:eye-acc}, \sieve, \swrd, \sgreedy\ and \snaive\ generate results of almost the same quality as the one given by \greedy, and \random\ gives the worst results in all selected windows. In both Figure \ref{fig:gas-acc} and Figure \ref{fig:wc-acc}, results generated by \swrd, \snaive, \sgreedy\ and \sieve\ are slightly worse than the one given by \greedy.  In most windows, \sgreedy\ is as good as \sieve. \snaive\ also performs well in most windows, but it is worse than \random\ in some windows. In theory, \swrd\ can be worse than \sieve\ by a factor of $2$, but our experiments show that solutions returned by the two algorithms have similar utilities.
Figure \ref{fig:wc-sgreedy-c} shows in \sgreedy, 
increasing $c$ will slightly increase the utility. 



For the comparisons of space/time costs, we only include the results of \eye\ dataset due to the space constraints. Similar results can observed on other datasets as well. Figure \ref{fig:eye-time} compares the time costs on \eye\ dataset. We measure the time costs by the numbers of function calls (of the submodular function). All results are normalized by the corresponding costs of \sieve.  
By Theorem \ref{thm:reduction} the time cost of \swrd\ is independent of $k$ and $W$ once it is normalized by the corresponding cost of \sieve. This result has been confirmed by Figure \ref{fig:call-rd}. Figure \ref{fig:call-snaive} shows that \snaive\ is as fast as \sieve. Figure \ref{fig:call-sgreedy} shows that increasing $k$ will increase the cost of \sgreedy, while increasing $W$ will decrease the cost. This is because items in the solution buffers are less likely to expire for small $k$ and large $W$.
Figure \ref{fig:call-sgreedy-c} shows how the time costs of \sgreedy\ are affected by the values of $c$.




Figure \ref{fig:eye-space} compares the space costs on \eye\ dataset. To be consistent with the theorems, we measure the space usages by the maximum numbers of items kept in memory. To compare with the costs of \sieve, we also normalize the costs of \swrd, \snaive\ and \sgreedy\ by the corresponding costs of \sieve. 
Figure \ref{fig:space-snaive-nm} and Figure \ref{fig:space-sgreedy-nm} show that the space usages of \snaive\ and \sgreedy\ are almost the same as \sieve. 


\paragraph{Summary}
We conclude from our experiments that (1) the distribution of data stream changes over sliding windows in our tested datasets; (2) in terms of solution quality, \swrd, \snaive\ and \sgreedy\ generate comparable results as \sieve, and \random\ is clearly the worst. \snaive\ can sometimes perform very badly, while \swrd\ (the only algorithm with theoretical guarantees) and \sgreedy\ are relatively stable; and (3) \snaive\ is the most time and space efficient algorithm among \swrd, \snaive\ and \sgreedy, and the performance of \sgreedy\ is close for large window size and small $k$. For large window size and small $k$, \sgreedy\ runs very fast and the only extra space it uses compared with \sieve\ is the buffer of samples (i.e. $B$). Depending on the value of $\eps^{-1}\log M$, \swrd\ typically uses $10$-$20$x processing time and $10$-$20$x  space compared to \sieve.




\clearpage
\bibliography{paper}

\begin{thebibliography}{10}

\bibitem{wc98}
M.~Arlitt and T.~Jin.
\newblock 1998 world cup web site access logs.
\newblock \url{http://ita.ee.lbl.gov/html/contrib/WorldCup.html}, 1998.

\bibitem{BDM02}
B.~Babcock, M.~Datar, and R.~Motwani.
\newblock Sampling from a moving window over streaming data.
\newblock In {\em SODA}, pages 633--634. Society for Industrial and Applied
  Mathematics, 2002.

\bibitem{BDMO03}
B.~Babcock, M.~Datar, R.~Motwani, and L.~O'Callaghan.
\newblock Maintaining variance and k-medians over data stream windows.
\newblock In {\em PODS}, pages 234--243. ACM, 2003.

\bibitem{BMKK14}
A.~Badanidiyuru, B.~Mirzasoleiman, A.~Karbasi, and A.~Krause.
\newblock Streaming submodular maximization: Massive data summarization on the
  fly.
\newblock In {\em SIGKDD}, pages 671--680. ACM, 2014.

\bibitem{B14}
J.~Bilmes.
\newblock Lecture notes on submodular optimization.
\newblock \url{http://j.ee.washington.edu/~bilmes/classes/ee596b_spring_2014/},
  2014.

\bibitem{BLLM15}
V.~Braverman, H.~Lang, K.~Levin, and M.~Monemizadeh.
\newblock A unified approach for clustering problems on sliding windows.
\newblock {\em arXiv preprint arXiv:1504.05553}, 2015.

\bibitem{BO07}
V.~Braverman and R.~Ostrovsky.
\newblock Smooth histograms for sliding windows.
\newblock In {\em FOCS}, pages 283--293, 2007.

\bibitem{CK15}
A.~Chakrabarti and S.~Kale.
\newblock Submodular maximization meets streaming: matchings, matroids, and
  more.
\newblock {\em Math. Program.}, 154(1-2):225--247, 2015.

\bibitem{CGQ15}
C.~Chekuri, S.~Gupta, and K.~Quanrud.
\newblock Streaming algorithms for submodular function maximization.
\newblock In {\em ICALP}, pages 318--330, 2015.

\bibitem{FSH+15}
J.~Fonollosa, S.~Sheik, R.~Huerta, and S.~Marco.
\newblock Reservoir computing compensates slow response of chemosensor arrays
  exposed to fast varying gas concentrations in continuous monitoring.
\newblock {\em Sensors and Actuators B: Chemical}, 215:618--629, 2015.

\bibitem{GKT12}
J.~Gillenwater, A.~Kulesza, and B.~Taskar.
\newblock Near-optimal map inference for determinantal point processes.
\newblock In {\em NIPS}, pages 2735--2743, 2012.

\bibitem{KG10}
A.~Krause and R.~G. Gomes.
\newblock Budgeted nonparametric learning from data streams.
\newblock In {\em ICML}, pages 391--398, 2010.

\bibitem{KMVV15}
R.~Kumar, B.~Moseley, S.~Vassilvitskii, and A.~Vattani.
\newblock Fast greedy algorithms in mapreduce and streaming.
\newblock {\em ACM Transactions on Parallel Computing}, 2(3):14, 2015.

\bibitem{LSH03}
N.~Lawrence, M.~Seeger, and R.~Herbrich.
\newblock Fast sparse gaussian process methods: The informative vector machine.
\newblock In {\em NIPS}, number EPFL-CONF-161319, pages 609--616, 2003.

\bibitem{Muthukrishnan05}
S.~Muthukrishnan.
\newblock {\em Data streams: Algorithms and applications}.
\newblock Now Publishers Inc, 2005.

\bibitem{NWF78}
G.~L. Nemhauser, L.~A. Wolsey, and M.~L. Fisher.
\newblock An analysis of approximations for maximizing submodular set
  functions—i.
\newblock {\em Mathematical Programming}, 14(1):265--294, 1978.

\bibitem{eye13}
O.~Roesler.
\newblock Eeg eye state data set.
\newblock
  \url{http://archive.ics.uci.edu/ml/machine-learning-databases/00264/}, 2013.

\bibitem{arma10}
C.~Sanderson.
\newblock Armadillo: An open source c++ linear algebra library for fast
  prototyping and computationally intensive experiments.
\newblock 2010.

\bibitem{SS02}
B.~Sch{\"o}lkopf and A.~J. Smola.
\newblock {\em Learning with kernels: support vector machines, regularization,
  optimization, and beyond}.
\newblock MIT press, 2002.

\bibitem{SSSJ12}
R.~Sipos, A.~Swaminathan, P.~Shivaswamy, and T.~Joachims.
\newblock Temporal corpus summarization using submodular word coverage.
\newblock In {\em CIKM}, CIKM '12, pages 754--763, 2012.

\bibitem{VVS06}
S.~Van~Vaerenbergh, J.~Via, and I.~Santamar{\'\i}a.
\newblock A sliding-window kernel rls algorithm and its application to
  nonlinear channel identification.
\newblock In {\em ICASSP}, volume~5, pages 789--792. IEEE, 2006.

\end{thebibliography}
\bibliographystyle{abbrv}


\end{document}